\documentclass[10pt,journal,compsoc]{IEEEtran}
\usepackage{cite}
\usepackage{algpseudocode}
\ifCLASSINFOpdf
\else
 \fi
\usepackage[cmex10]{amsmath}
\usepackage{graphics} % for pdf, bitmapped graphics files
\usepackage{epsfig} % for postscript graphics files
\usepackage{amssymb}  % assumes amsmath package installed
\usepackage{amsthm}
\usepackage{algpseudocode}
\usepackage[ruled]{algorithm}
\newcommand{\tb}{\textbf}
\newcommand{\mc}{\mathcal}
\usepackage{multicol}

\begin{document}

\title{Weighted network estimation by the use of topological graph metrics}

\author{Loukianos Spyrou and Javier Escudero\thanks{This work was supported by EPSRC, UK, Grant No. EP/N014421/1}
\thanks{Loukianos Spyrou and Javier Escudero are with the School of Engineering, University of Edinburgh, EH9 3FB, U.K.}% <-this % stops a space
}
\twocolumn

\markboth{}%
{}

\IEEEtitleabstractindextext{%
\begin{abstract}
Topological metrics of graphs provide a natural way to describe the prominent features of various types of networks. Graph metrics describe the structure and interplay of graph edges and have found applications in many scientific fields. In this work,  graph metrics are used in network estimation by developing optimisation methods that incorporate prior knowledge of a network's topology. The derivatives of graph metrics are used in gradient descent schemes for weighted undirected network denoising, network completion , and network decomposition. The successful performance of our methodology is shown in a number of toy examples and real-world datasets. Most notably, our work establishes a new link between graph theory, network science and optimisation. 
\end{abstract}

}

\maketitle
\IEEEpeerreviewmaketitle
\IEEEraisesectionheading{\section{Introduction}\label{sec:introduction}}

\IEEEPARstart{G}{raph} theory has found applications in many scientific fields in an attempt to analyse interconnections between phenomena, measurements, and systems. It provides a data structure that naturally express those interconnections and also provides a framework for further analysis \cite{bollobas2012graph,gross2005graph}. A graph consists of a set of nodes and edges describing the connections between the nodes. The edges of binary graphs take the values of either 1 or 0 indicating the presence or absence of a connection, while in weighted graphs the edges are described by weights indicating the strength of the connection. Graphs have been extensively used in a variety of applications in network science such as biological networks, brain networks, and social networks \cite{deo2016graph,barabasi2012network,tichy1979social,girvan2002community}.

Graph metrics are functions of a graph's edges and characterize one or several aspects of network connectivity \cite{Barrat2004,Rubinov2010,Laita2011}. Local metrics deal with the relation of specific nodes to the network structure while global metrics describe properties of the whole network. Graph metrics are largely used to describe the functional integration or segregation of a network, quantify the centrality of individual regions, detect community structure, characterize patterns of interconnections, and test resilience of networks to abrupt changes.

Estimation of a network's structure or properties has been performed in a variety of contexts. Although a single definition does not exist, network estimation can include any algorithm or method that detects, enhances, generates, or increases some quality measure of networks. Link prediction and network completion deal with predicting the existence of missing edges based on the observed links in a binary graph \cite{Liben-Nowell2007,Goldberg2003,Lu2010,Kim2011,Hanneke2009}. Also, in \cite{Kim2011} the prediction of missing nodes was attempted. So far such methods have dealt with detecting only whether a link (edge) exists or not, not with the estimation of its weight. Typical applications include predicting the appearance of future connections in social \cite{Liben-Nowell2007,al2006link} or biological \cite{symeonidis2013biological} networks. The network reconstruction problem deals with composing networks that satisfy specific properties \cite{Mastrandrea2014,Squartini2011,Bleakley2007,Filosi2014}. This can be particularly useful when building null models for hypothesis testing. Network inference problem attempts to identify a network structure where the edges have been corrupted by a diffusion process through the network \cite{Gomez-Rodriguez2010,myers2010convexity}. More recently, reducing the noise in social networks has been attempted in \cite{Aghagolzadeh2015,Morris2004,Gao2013}.

In this work, we assume that prior information is available regarding an observed weighted undirected network. This prior information comes in the form of estimates of a number of topological graph metrics of the network. We utilise these estimates in an optimisation framework in order to adjust the weights of the observed network to satisfy those properties. There are many real-world cases where there is knowledge of a network's structure but not exact or reliable network information \cite{Mastrandrea2014}, e.g. strength of connections between banks is known but exact connections are hidden for privacy issues in bank networks \cite{Caldarelli2013}, modularity metrics of brain networks are similar between subjects \cite{Sporns2016}, properties of the constituent networks may be known for mixed networks \cite{Morris2004}. We demonstrate the utility of our methodology in three schemes. 

Firstly, a network denoising scheme, where an observed network is a noisy version of an underlying noise-free network. By considering the error between the observed network's graph metrics and the true network's known metrics, in an iterative gradient descent process on the network's weights, the resulting network is closer to the noise-free one. Using the node degrees as priors has been performed in binary networks in \cite{Morris2004} and in terms of network reconstruction the knowledge of degrees has been employed in weighted networks in \cite{Mastrandrea2014,Huang2009}. In \cite{Aghagolzadeh2015}, the transitivity of a network was used in a network diffusion scheme to change a social network's weights but without a specific stopping point. In \cite{Gao2013}, denoising has been attempted in the context of removing weak ties between users in a social networks. Here, we provide analytical and empirical proofs on the utility of denoising schemes that are based on the optimisation of various graph metrics through gradient descent.

Secondly, we develop a weighted network completion scheme where some of the weights of the network's edges are missing. Similarly to the two previous schemes, we adapt the missing weights such that the whole network obtains specific values for known graph metrics. Weighted network completion has not been performed in the literature per se, only the closely related matrix completion problem \cite{Keshavan2010} and matrix completion on graphs \cite{Kalofolias2014} where the completion is aided by assuming that the rows and columns of the matrix form communities.

Finally, we develop a network decomposition scheme for the cases where an observed network is an additive mixture of two networks. Assuming that graph metrics are known for the two constituent networks we derive an algorithm that can estimate the networks by enforcing them to have specific values for graph metrics while keeping the reconstruction error between the original and estimated mixture small. Network decomposition has been traditionally applied in a different context, on decomposing graphs with disjoint nodes \cite{Nash-Williams1964}. For factored binary graphs untangling has been performed by taking into account the degree distribution of the constituent graphs \cite{Morris2004}. Here, we not only consider the degrees of a network and decomposing it into subgraphs (i.e. multiple graphs with disjoint nodes) but we use multiple graph metrics in additive graph mixtures.

Therefore, in this paper, we provide a comprehensive description of theoretical and empirical results on the use of optimisation of graph metrics for various problems in network estimation. In section \ref{sec:graph_metrics} we give some basic definitions and a brief introduction to graph theory and graph metrics. The network estimation methods are shown in Section \ref{sec:estimation} with the details of the three schemes, denoising \ref{sec:denoising}, completion \ref{sec:completion}, and decomposition \ref{sec:decomposition}. In section \ref{sec:deriv} we derive the graph metrics derivatives that are used in the optimisation methods. In section \ref{sec:results} we apply our methodology to a number of toy examples and real data and in section \ref{sec:discussion} we put the results into context and discuss the utility that our method provides. Section \ref{sec:conclusions} concludes the paper.

\section{Graph metrics} \label{sec:graph_metrics}

A weighted graph $\mathcal{G}= (\mathcal{V}, \mathcal{E}, \bf W)$ is defined by a finite set of nodes (or vertices) $\mc{V}$ with $|\mc{V}|=n$, a set of edges $\mc{E}$ of the form $(v_i,v_j) \in \mc{E}$ with $|\mc{E}|=n^2-n$ and a weighted adjacency matrix $\bf{W}$ with $w_{ii}=0 \,\forall\, i$. In this work, we consider undirected graphs for which $\bf{W}$ is symmetric, i.e. $w_{ij}=w_{ji}$ or $\bf{W}=\bf{W^T}$. The entries $w_{ij}$ in the weighted adjacency matrix $\bf W$ (weight matrix from now on) indicate the strength of connection between nodes. We assume that networks are normalised, i.e. $w_{ij}\in [0,1]$.

Graph metrics are scalar functions of the weight matrix, i.e. $f(\tb{W}): \mathbb{R}^{n^2} \to \mathbb{R}$. Global metrics map the weight matrix into a single value and therefore attempt to simply quantify a specific property of a network. Local metrics on the other hand, quantify some property of the network separately for each node $i \in \{1...n\}$ with $f_i(\tb{W}) \in \mathbb{R}^{n^2} \to \mathbb{R}$, potentially resulting in $n$ functions and $n$ separate values.

Although graph metrics were originally defined on binary (unweighted) networks, the conversion to weighted metrics is usually but not always straightforward \cite{Rubinov2010,Barthelemy2005,Newman2004,Fagiolo2004,Onnela2005,Opsahl2009}. 
The main motivation of this study can be recognised by pointing out that a network can be adjusted by changing the matrix $\tb{W}$ so that it obtains a certain value for some graph measure.
There are numerous graph metrics that describe various features of networks  \cite{Rubinov2010,Laita2011}. The main properties that they measure are:
\begin{itemize}
	\vspace{5pt}
\item Integration (ability of the network to combine information from distributed nodes).
\item Segregation (ability of the network for specialised processing in densely interconnected node groups).
\item Centrality (characteristic that describes the presence of regions responsible for integrating total activity).
\item Resilience (ability of a network to withstand weight changes).
\item Motifs (presence of specific network patterns).
\item Community structure (separation of a network to various functional groups).
\item Small-worldness (highly segregated \emph{and} integrated nodes).
\vspace{5pt}
\end{itemize}

There are graph metrics that contain quantities which are themselves a product of an optimisation procedure (e.g. module structure, shortest path length). These quantities are considered outside of the scope of this work. 

\subsection{Definitions}
 Here we show some definitions, notations and useful matrices that are used in the following sections.
 \\ 
 $\begin{array}{ll}
 \hline \hline \rule{0pt}{4ex}
 n & \text{:\quad number of nodes} \\\\
 m & \text{:\quad number of graph metrics }\\\\
 {\tb{S}_{ij}}=\{1_{ij}\}, \{0_{\neg i \| \neg j}\} & \text{:\quad matrix of zeros except at } (i,j)\\\\
 \{\tb{A}\}_{ij}=a_{ij}=tr\{\tb{A}\tb{S}_{ji}\}  & \text{:\quad $(i,j)_{th}$ element of matrix $\bf{A}$}\\\\
 tr\{\tb{A}\}=\sum\limits_{i=1}^{n} \{\tb{A}\}_{ii}	&\text{:\quad sum of diagonal elements} \\\\
\tb{1}_n & \text{:\quad column vector of ones}\\\\
 \tb{O}_n=\tb{1}_{n}\tb{1}_{n}^T	& \text{:\quad matrix of ones} \\\\
 \tb{H}_n=\tb{O}_n-\tb{I}_n & \text{:\quad all ones except at diagonal} \\\\
 {\tb{R}_j}=\{1_{ij} \forall i\},\{0_{\neg i j \forall i}\} & \text{:\quad matrix with ones at column $j$}\\\\
 \sum\limits_{ij}a_{ij}=tr\{\tb{A}\tb{O}_n\} & \text{:\quad sum of all matrix entries}\\\\
 \tb{A} \circ \tb{B} & \text{:\quad Hadamard (element-wise) product}\\\\
 \hline \rule{0pt}{-2ex}
 \end{array} $

 \section{Network Estimation}\label{sec:estimation}
 
 In this section we formulate the optimisation methodologies for the three schemes considered in this work.

 \subsection{Denoising}\label{sec:denoising}
 Suppose that a weight matrix $\bf W$ of a network is corrupted by additive noise:
 
 \begin{equation}
 \tb{W}_e =\bf W + E
 \end{equation}
 
 The error matrix $\bf E$ can be considered as a network unrelated to the network structure being considered. For example, $\bf E$ could be social calls when trying to detect suspicious calls in social networks \cite{Morris2004}, the effect of volume conduction in EEG based brain network connectivity, or measurement noise. A different type of noise that occurs in networks, the effect of missing values is treated in section \ref{sec:completion}.
 
 If we assume that we have estimates of $M$ differentiable graph metrics of the original $\bf W$, i.e. $f_m({\bf W})=K_m$ where $m\in \{1,... ,M\}$, then we can formulate a cost function that measures the deviation of the observed weight matrix's metrics $f_m(\bf W_e)$ to the estimates $K_m$ as:
 \begin{equation}
 c(\tb{W}_e)=\sum\limits_me_m^2(\tb{W}_e)=\sum\limits_m\left(f_m({\bf W}_e)-K_m\right)^2
 \end{equation}
 The error is minimised with gradient descent updates on $\bf W_e$:
 \begin{equation}\label{denoising_updates}
 {\bf W}_e^{(t+1)}={\bf W}_e^t- \mu \sum\limits_m  e_m({\bf W}_e^t)\frac{df_m({\bf W}_e^t)}{{d\bf W}_e^t}
 \end{equation}
 where $t$ is the iteration index and $\mu$ the learning rate. For the case of $m=1$, Equation (\ref{denoising_updates}) describes the traditional single function gradient descent. For $m>1$, it can be considered an equally weighted sum method of multiobjective optimisation, motivated by the fact that the graph metrics are in the same range due to normalisation. Multiobjective optimisation enables the weighting of different metrics to accommodate priorities on which are more important for a specific task. Such weighting is considered above the scope of this work. The full denoising procedure is described in Algorithm \ref{alg:denoising}. Note that values below 0 and above 1 are truncated to zero and one respectively since we assume that the networks are normalised.
 
 In Appendix A we provide a proof that for convex cost functions $c(\tb{W})$, denoising guarantees error reduction. The implication of that is that when a graph metric results in a convex cost function, such as the degrees ($k_i^w$) of the network, then the optimisation of Algorithm \ref{alg:denoising} with $f(\tb{W})=\frac{1}{n}\sum\limits_i k_i^w$ will always converge to a solution $\hat{\bf W}$ that is closer to the original network $\tb{W}$ than $\tb{W}_e$ is. For non convex metrics, there is no such guarantee. In Section \ref{sec:denoising} we show empirical results on the extent of that effect. In the Appendix C we show the proof that cost functions based on graph metrics such as the degree are convex. In general, linear functions of the weights of the adjacency matrix are convex whereas nonlinear functions of the weights are not.

 \begin{algorithm}
 	\begin{algorithmic}[1]
 		\Statex OUTPUT:  $\hat{\bf W}$\Statex INPUTS: $\bf W_e$, $K_m$
 		\State Initialise $t=0, {\bf W}^0={\bf W}_e$, 
 		\State $E$=$\sum\limits_me_m^2({\bf W}^0)=\sum\limits_m\left(f_m({\bf W}^0)-K_m\right)^2$
 		\While{$E>\epsilon$}  
 		\State ${\bf W}^{(t+1)}={\bf W}^t- \mu \sum\limits_m  e_m({\bf W}^t)\frac{df_m({\bf W}^t)}{{d\bf W}^t}$
 		\State if $w_{ij}<0$ then $w_{ij}=0$, $w_{ij}>1$ then $w_{ij}=1$
 		\State $E$=$\sum\limits_me_m^2({\bf W}^{t+1})=\sum\limits_m\left(f_m({\bf W}^{t+1})-K_m\right)^2$
 		\State $t=t+1$
 		\EndWhile
 		\State Denoised network: ${\bf \hat{W}}={\bf W}^t$
 	\end{algorithmic}
 	\caption{Denoising of corrupted network ${\bf W}_e$ through the use of graph measure estimates $K_m$}\label{alg:denoising}
 \end{algorithm}
 
 \subsection{Completion} \label{sec:completion}
 
 For the case that a set $\mathcal{M}$ of entries of the weight matrix are missing, we can perform matrix completion by:
 \begin{equation}\label{completion_updates}
 {\bf W}_{ic}^{(t+1)}={\bf W}_{ic}^t- \mu \sum\limits_m  e_m({\bf W}_{ic}^t)\left( \frac{df_m({\bf W}_{ic}^t)}{{d\bf W}_{ic}^t}\circ {\bf S_\mathcal{M}}\right)
 \end{equation}
 where ${\bf S}_\mathcal{M}$ is a matrix with ones at the set of missing entries and zeroes everywhere else. Similar to the previous two cases, the assumption is that if the true graph metrics are known, gradient descent will adjust the missing weights close to their true values. The missing entries of the incomplete weight matrix ${\bf W}_{ic}$ can be initialised to the most likely value of the network ($w$). This can be considered as a denoising procedure with the missing weights equal to `noisy' weights of value $w$. In Algorithm \ref{alg:completion} we describe the network completion procedure.
 
 \begin{algorithm}
 	\begin{algorithmic}[1]
 		\Statex OUTPUT:  $\hat{\bf W}_c$
 		\Statex INPUTS: ${\bf W}_{ic}$, $K_m$, set of missing entries $\mathcal{M}$
 		\State Initialise $t=0, \{{{\bf W}_{ic}^0}\}_{ij}=w \text{ with } (i,j) \in \mathcal{M}$, 
 		\State $E$=$\sum\limits_me_m^2({\bf W}_{ic}^0)=\sum\limits_m\left(f_m({\bf W}_{ic}^0)-K_m\right)^2$
 		\While{$E>\epsilon$}  
 		\State ${\bf W}_{ic}^{(t+1)}={\bf W}_{ic}^t- \mu \sum\limits_m  e_m({\bf W}_{ic}^t)\frac{df_m({\bf W}_{ic}^t)}{{d\bf W}_{ic}^t}\circ {\bf S_\mathcal{M}}$
 		\State if $w_{ij}<0$ then $w_{ij}=0$, $w_{ij}>1$ then $w_{ij}=1$
 		\State $E$=$\sum\limits_me_m^2({\bf W}_{ic}^{t+1})=\sum\limits_m\left(f_m({\bf W}_{ic}^{t+1})-K_m\right)^2$
 		\State $t=t+1$
 		\EndWhile
 		\State Complete network estimate: ${\bf \hat{W}}_c={\bf W}_{ic}^t$
 	\end{algorithmic}
 	\caption{Completing the missing entries of network $\tb{W}_{ic}$ through the use of graph measure estimates $K_m$}\label{alg:completion}
 \end{algorithm}
 
 \subsection{Decomposition} \label{sec:decomposition}
 
 Suppose that we observe a mixed network that arises as a combination of two networks:
 \begin{equation}
 \tb{W}_f = \tb{W}_1 + \tb{W}_2
 \end{equation}
 If we have estimates of some topological properties of the two networks, i.e. $\{f_m^1({\bf W}_1)=K^1_m\},\{f_m^2({\bf W}_2)=K^2_m\}$, then we can utilise these information to infer the networks from their mixture. This could be accomplished separately for each network using Algorithm \ref{alg:denoising}. However, since we know the mixture ${\bf W}_f$, we utilise that in the following optimisation problem:
 
 \begin{equation*}
 \begin{aligned}
 & \underset{\tb{W}_1,\tb{W}_2}{\operatorname{argmin}} 
 & & \sum\limits_m \left(f_m^1({\bf W}_1) - K^1_m\right)^2 + \left(f_m^2({\bf W}_2) - K^2_m\right)^2 \\
 & \text{subject to}
 & & ||{\bf W}_f-(\tb{W}_1+\tb{W}_2)||^2_{\text{F}}\leq \xi
 \end{aligned}
 \end{equation*}
 We solve this optimisation problem with alternating minimisation since it is a function of two matrices. We fix one of the two weight matrices and solve the following optimisation problem for the other one in an alternating fashion:
 \begin{equation}\label{eq:deccon1}
 \begin{aligned}
 & \underset{\tb{W}_1}{\operatorname{argmin}} 
 & & \sum\limits_m \left(f_m^1({\bf W}_1) - K^1_m\right)^2 +\lambda||\tb{W}_1-(\tb{W}_f-\tb{W}_2)||^2_{\text{F}}
 \end{aligned}
 \end{equation}
 \begin{equation}\label{eq:deccon2}
 \begin{aligned}
 & \underset{\tb{W}_2}{\operatorname{argmin}} 
 & & \sum\limits_m \left(f_m^2({\bf W}_2) - K^2_m\right)^2 +\lambda|| \tb{W}_2-(\tb{W}_f-\tb{W}_1)||^2_{\text{F}}
 \end{aligned}
 \end{equation}
 where the constraint has been incorporated into the cost function through the penalty parameter $\lambda$. Each separate minimisation, Algorithm \ref{alg:constrained}, resembles Algorithm \ref{alg:denoising} but in this case deviations from $\tb{W}_f-\tb{W}_{1}$ or $\tb{W}_f-\tb{W}_{2}$ are penalised. This whole procedure is shown in Algorithm \ref{alg:alternating}. 
 \begin{algorithm}
 	\begin{algorithmic}[1]
 		\Statex OUTPUT:  $\hat{\bf W}$
 		\Statex INPUTS:  $\tb{Y}$, $K_m$, $\lambda$
 		\State Initialise $t=0$, $\tb{W}^0=\tb{Y}$
 		\State $E$=$\sum\limits_me_m^2({\bf W}^0)=\sum\limits_m\left(f_m({\bf W}^0)-K_m\right)^2$
 		\While{$E>\epsilon$} 
 		\State ${\bf W}^{(t+1)}={\bf W}^t- \mu  \left(\sum\limits_m  e_m({\bf W}^t)\frac{df_m({\bf W}^t)}{{d\bf W}^t}+\lambda (\bf W-Y)||{\bf W-Y}||^2_{\text{F}}\right)$
 		\State if $w_{ij}<0$ then $w_{ij}=0$, $w_{ij}>1$ then $w_{ij}=1$
 		\State $E$=$\sum\limits_me_m^2({\bf W}^{t+1})=\sum\limits_m\left(f_m({\bf W}^{t+1})-K_m\right)^2$		
 		\State $t=t+1$
 		\EndWhile
 		\State Constrained estimate: $\hat{\tb{W}}=\tb{W}^t$
 	\end{algorithmic}
 	\caption{Constrained network optimisation based on graph metrics $K_m$ and a constraint that penalises deviations from a reference network $\tb{Y}$. The penalty is adjustable through the parameter $\lambda$.}\label{alg:constrained}
 \end{algorithm}
 
 \begin{algorithm}
 	\begin{algorithmic}[1]
 		\Statex OUTPUT:  $\hat{\tb{W}_1},\hat{\tb{W}_2}$
 		\Statex INPUTS: $\tb{W}_f$, $K_m^1$, $K_m^2$, $\lambda$
 		\State Initialise $t=0$, 
 		\State $\tb{W}_1^0=\text{Algorithm1$\{\tb{W}_f, K^1_m\}$}$	
 		\State $\tb{W}_2^0=\text{Algorithm1$\{\tb{W}_f, K^2_m\}$}$
 		\State $R=||{\bf W_m} - (\tb{W}_1^0+\tb{W}_2^0)||^2_{\text{F}}$
 		\While{$ R>\epsilon$}  
 		\State $\tb{W}_1^{t+1}=\text{Algorithm3}\{\tb{W}_f-\tb{W}_2^{t},\,\,\,\,\,K_m^1,\lambda\}$
 		\State if $w^1_{ij}<0$ then $w^1_{ij}=0$, $w^1_{ij}>1$ then $w^1_{ij}=1$
 		\State $\tb{W}_2^{t+1}=\text{Algorithm3}\{\tb{W}_f-\tb{W}_1^{t+1},K_m^2,\lambda\}$
 		\State if $w^2_{ij}<0$ then $w^2_{ij}=0$, $w^2_{ij}>1$ then $w^2_{ij}=1$
 		\State $R=||{\bf W_f} - (\tb{W}_1^{t+1}+\tb{W}_2^{t+1})||^2_{\text{F}}$
 		\State $t=t+1$
 		\EndWhile
 		\State Optimised estimates: $\tb{W}_1, \tb{W}_2$
 	\end{algorithmic}
 	\caption{Alternating minimisation procedure for network decomposition of a mixed network $\tb{W}_f$ through known graph metrics for the individual networks $K_m^1,K_m^2$.}\label{alg:alternating}
 \end{algorithm}
 
 The motivation for the procedure in Algorithm \ref{alg:alternating} is the following. Consider estimating the two weight matrices only by using the denoising algorithm for each one separately. The estimate of ${\tb{W}}_1$ can be written as $\hat{\tb{W}}_1=\tb{W}_1+\tb{E}_1$ where $\tb{E}_1$ indicates the error from the true weight matrix $\tb{W}_1$. Similarly for $\tb{W}_2$ resulting in $\hat{\tb{W}}_f=\tb{W}_1+\tb{W}_2+\tb{E}_1+\tb{E}_2=\tb{W}_f+\tb{E}$. Therefore, it is evident that the estimates of the weight matrices are not ideal whenever $||\tb{E}||_{F}>0$. Note that if $||\tb{E}||_{F}=0$ it does not necessarily imply that the estimates of the weight matrices are optimal since it is possible that $\tb{E}_1=-\tb{E}_2$. It would be optimal if we could constraint the estimate e.g. $\hat{\tb{W}}_1$ as:
 
 \begin{equation*}
 \begin{aligned}
 & \underset{\tb{W}_1}{\operatorname{argmin}} 
 & & \sum\limits_m \left(f_m^1({\bf W}_1) - K^1_m\right)^2 \\
 & \text{subject to}
 & & ||{\bf E}_1||^2_{\text{F}}=0
 \end{aligned}
 \end{equation*}
 However that would require knowledge of the true weight matrix $\tb{W}_1$. Instead, note that: \\ $\hat{\tb{W}}_1-(\tb{W}_f-\hat{\tb{W}}_2)=\tb{W}_1+\tb{E}_1-(\tb{W}_1+\tb{W}_2-\tb{W}_2-\tb{E}_2)
 =\tb{E}_1+\tb{E}_2$. Therefore by reducing $||\tb{W}_1-(\tb{W}_f-\tb{W}_2)||_{F}$ in Eq. (\ref{eq:deccon1}), we are reducing the total error. In other words we are solving the following constrained optimisation problem:
 \begin{equation}
 \begin{aligned}
 & \underset{\tb{W}_1}{\operatorname{argmin}} 
 & & \sum\limits_m \left(f_m^1({\bf W}_1) - K^1_m\right)^2 \\
 & \text{subject to}
 & & ||{\bf E}_1 + \tb{E}_2||^2_{\text{F}}\leq \xi
 \end{aligned}\label{eq:explanation}
 \end{equation}
 The inequality is incorporated such that $||\tb{E}_2||_{F}\geq0$ when optimising $\tb{W}_1$ and vice versa. Consider the extreme cases. Firstly, when $\xi=0$. That would imply that $\lambda=\infty$ which would render the graph measure optimisation ineffective. On the other hand, a large $\xi$ implies small $\lambda$ which would render the constraint ineffective. In our implementation we adjust the $\lambda$ parameter such that the reconstruction error, i.e. $\tb{W}_f-\hat{\tb{W}}_f$ slowly decreases over the iterations of the alternating minimisation algorithm. Note the well known fact that there is a one-to-one correspondence between $\xi$ in Eq. \ref{eq:explanation} and $\lambda$ in Eq. \ref{eq:deccon1}.

\subsection{Derivatives of graph metrics} \label{sec:deriv}
 
In this section we derive the expressions for the derivatives of popular graph metrics that describe some important properties of networks. We deal with: degree, average neighbour degree, transitivity, clustering coefficient, modularity. More details can be found in Appendix B.
\\
\subsubsection{Degree}
The degree of a node $i$ describes the connection strength of that node to all other nodes: 
\begin{equation}
k_i^w=\sum\limits_{j}w_{ij}=tr\{\tb{W} \tb{R}_i\}
\end{equation}
with the degree derivative being:
\begin{equation}
\frac{\partial k_i^w}{\partial {\bf W}}={\bf R}_i^T
\end{equation}
Since $\tb{R}_i$ is non-zero only for column $i$ it can be computed efficiently as:
\begin{equation}
\frac{\partial k_i^w}{\partial {\bf w}_i}={\tb{1}_n}^T
\end{equation}
where $\tb{w}_i$ is the $i_{th}$ column of $\tb{W}$.
\\
\subsubsection{Average neighbour degree - resilience}
The average neighbour degree for node $i$ is given by:
\begin{equation}
ND_i=\frac{\sum\limits_j w_{ij}k_j^w}{k_i^w}=\frac{tr\{\tb{W}^2\tb{R}_i\}}{tr\{\tb{W}\tb{R}_i\}}=\frac{\rho}{\tau}
\end{equation}
The derivative of the average neighbour degree is:
\begin{equation}
\frac{\partial ND_i}{\partial {\tb{W}}}=\frac{\tau(\tb{W}\tb{R}_i+\tb{R}_i\tb{W})^T-\rho \tb{R}_i^T}{\tau^2}
\end{equation}
\\
\subsubsection{Transitivity - segregation}

The transitivity is a global measure of the segregation of a network and here we defined it as (see also \cite{Aghagolzadeh2015}):
\begin{equation}\label{trans}
T=\frac{\sum\limits_{ijh}w_{ij} w_{ih}w_{jh}}{\sum\limits_{ij}\sum\limits_{h}w_{ih}w_{jh}}=\frac{tr \bf \{W^3\}}{ tr  \{\tb{W}\tb{H}_n\tb{W}\}}=\frac{\alpha}{\beta}
\end{equation}

The transitivity derivative is:
\begin{equation}
\frac{\partial T}{\partial \bf W} = \left(\frac{3\beta \tb{W}^2-\alpha(\tb{W}\tb{H}_n+\tb{H}_n\tb{W})}{\beta^2}\right)
\end{equation}
\\\\
\subsubsection{Clustering coefficient - segregation}
The clustering coefficient for node $i$ is a local measure of the clustering of a network. It is defined as:
\\
\begin{equation}\label{clusti}
C_i=\frac{\sum\limits_{jh}w_{ij}w_{ih}w_{jh}}{\sum\limits_{jh} w_{ij}w_{ih}}=\frac{\{\tb{W}^3\}_{ii}}{\{{\bf WH_nW}\}_{ii}}=\frac{tr\{\tb{S}_{ii} \tb{W}^3\}}{tr\{\tb{S}_{ii} \tb{W}\tb{H}_n\tb{W}\}}=\frac{\gamma_i}{\zeta_i}
\end{equation}
\\
The derivative of the clustering coefficient is:
\\
\begin{align*}
&\frac{\partial C_i}{\partial \bf W}  = \\ &  \left( \frac{3\zeta_i  \sum\limits_r^{2}\left(\tb{W}^r \tb{S}_{ii} \tb{W}^{2-r}\right)-\gamma_i( \tb{S}_{ii}^T \tb{W}^T \tb{H}^T+\tb{H}^T\tb{W}^T\tb{S}^T_{ii})}{\zeta_i^2}\right)
\end{align*}
\\
\subsubsection{Modularity - community structure}
Modularity metrics the tendency of a network to be divided into modules. Here we deal with optimising the modularity in terms of the network weights, not in terms grouping nodes into modules. Modularity can be written as (see also \cite{Chang2013}):
\\
\begin{equation}
M=\frac{1}{l^w}\sum\limits_{ij}\left(w_{ij}-\frac{k_i^w k_j^w}{l^w}\right)\delta_{ij}
\end{equation}
\\
where $\delta_{ij}=1$ whenever nodes $i$ and $j$ belong to the same module and zero otherwise. 

The modularity derivative is expressed as:
\begin{equation}
\frac{\partial M}{\partial {\bf W}}=\frac{\partial m_1}{\partial \bf W}-\frac{\partial m_2}{\partial \bf W}
\end{equation} 
with:
\\
\begin{equation}
\frac{\partial m_1}{\partial \bf W}=\frac{ l^w{\bf\Delta}-{\theta \bf O_n^T} }{(l^w)^2}
\end{equation}
\\
and:
\\
\begin{equation}
\frac{\partial m_2}{\partial \bf W}=\sum\limits_{r=1}^n\frac{l^w({\bf C_r W \Delta^T+ C_r^T W \Delta})-2 \xi_r {\bf O_n^T}}{(l^w)^3}
\end{equation} 
\\
where $\bf{C_r}$ is a circular shift matrix that shifts down the rows of the matrix on the right by $r-1$ and $\xi_r=\bf{W^T}\bf{C}_r\bf{W}\Delta^T $. See Appendix B.B for more details.

\subsubsection{Local and global metrics}
Any graph measure $f_i$ that operates locally on node $i$ can be cast into its global (full network) form by evaluating the gradient as the average of the nodes' derivatives:
\begin{equation}
\frac{\partial f}{\partial {\bf W}}=\frac{1}{n}\sum\limits_i \frac{\partial f_i}{\partial {\bf W}}
\end{equation}
\section{Results}\label{sec:results}

\subsection{Denoising}

\subsubsection{Synthetic Networks}

In this section we show results of applying the denoising algorithm for various cases. 
We create synthetic undirected networks of three types: 
\begin{itemize}
	\item Random complete network , $w_{ij} \sim \mc{U}[0,1] $
	\item Scale free weighted network where the degrees are distributed with the power law and the non zero edges are given weights $w_{ij} \sim \mc{U}[0,1]$. The network was created according to \cite{Prettejohn2011} with average degree of 5.
	\item A modular network where the weights exhibit community structure in a number of modules. The network was created with the BCT toolbox \cite{Rubinov2010}. The network consists of 8 modules and $90\%$ of the non-zero weights in the modules. 
\end{itemize}  
For each case we add a noise matrix where each entry of $\tb{E}$, $e_{ij} \sim \mc{N}(0,1)$ is normally distributed with mean 0 and standard deviation 1: 

\begin{equation}
\tb{W}_e=\tb{W}+\sigma \tb{E}
\end{equation}

Weights of $\tb{W}_e$ that go below 0 are set to zero, and subsequently the weight matrix is normalised by dividing by its maximum value. In that way we guarantee that all elements of $\tb{W}_e$ are between 0 and 1.
%
%Firstly, for various graph metrics we show examples of gradient trajectories for a single noise realisation of the random complete network. The metric used is the distance between the estimated and true weight matrix. The results are shown in Figure \ref{fig:denconv1}.
%\begin{figure}[htbp]
%	\centering    
%	\includegraphics[width=9cm]{denconv1.pdf}
%	\caption{Examples of convergence for various graph metrics for a single noise realisation ($\sigma=0.8$). We show the error of the estimated network w.r.t. the true network in terms of the algorithm's iterations. The network here is a 128 node random complete network. The main point of this figure is to show that as the algorithm was converging to the required graph measure value, for the transitivity and clustering coefficient, the error was increased.}
%	\label{fig:denconv1}
%\end{figure}

Firstly, we show the error reduction of the scheme for various noise levels and networks of 128 nodes. We define error reduction as the ratio of the error of the denoised network $\hat{\tb{W}}$ to the error of the noisy network $\tb{W}_e$:
\begin{equation}\label{eq:erred}
er=1-\frac{||\hat{\tb{W}}-\tb{W}||_{\text{F}}}{||\tb{W}_e-\tb{W}||_{\text{F}}}
\end{equation}
Error reduction is measure in the domain $(-\infty, 1]$ where $1$ indicates perfect denoising. Negative values indicate larger error after applying the denoising algorithm.

In Figures \ref{fig:denrandomvn}, \ref{fig:denscalevn} and \ref{fig:denmodvn} we show the error reduction for an increasing noise level and different graph metrics for the random, scale-free and modular network respectively. Each noise level considers the average of 50 noise matrix realisations. Note that even though the error reduction increases as the noise increases, in absolute terms the error always increases.

\begin{figure}[htbp]
	\centering    
	\includegraphics[width=8cm]{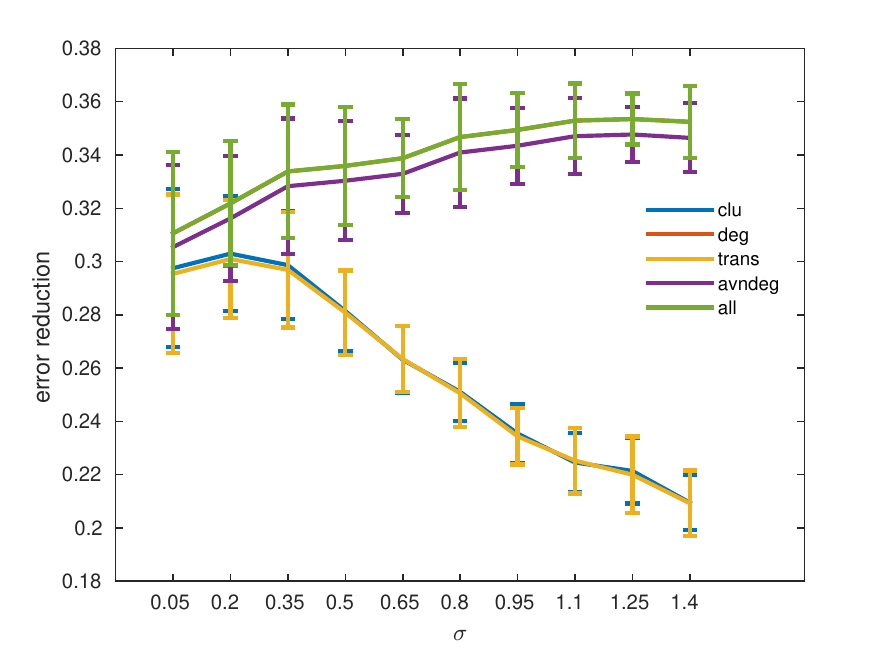}
	\caption{Denoising of a random complete network with 128 nodes. We show the error reduction for an increasing noise level and various graph metrics. Note that the degree coincides with the combination of graph metrics in this figure. This indicates that for the cases that the degree is known no other network metrics are necessary.}
	\label{fig:denrandomvn}c
\end{figure}

\begin{figure}[htbp]
	\centering    
	\includegraphics[width=8cm]{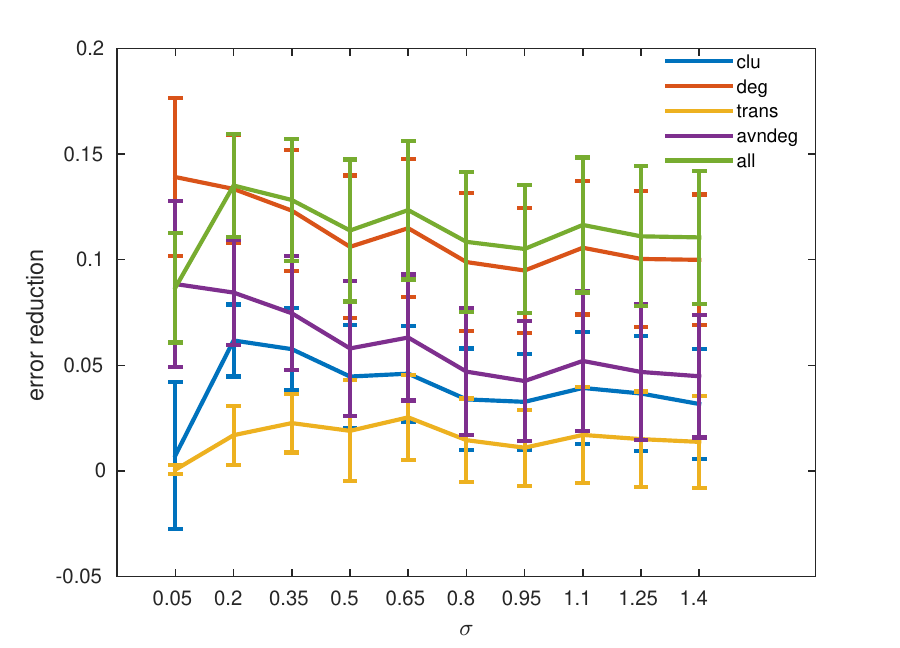}
	\caption{Denoising of a scale-free network with 128 nodes. We show the error reduction for an increasing noise level and various graph metrics. Clustering based measures suffer in this type of network due to the decreased magnitude of those measures in scale-free networks.}
	\label{fig:denscalevn}
\end{figure}

\begin{figure}[htbp]
	\centering    
	\includegraphics[width=8cm]{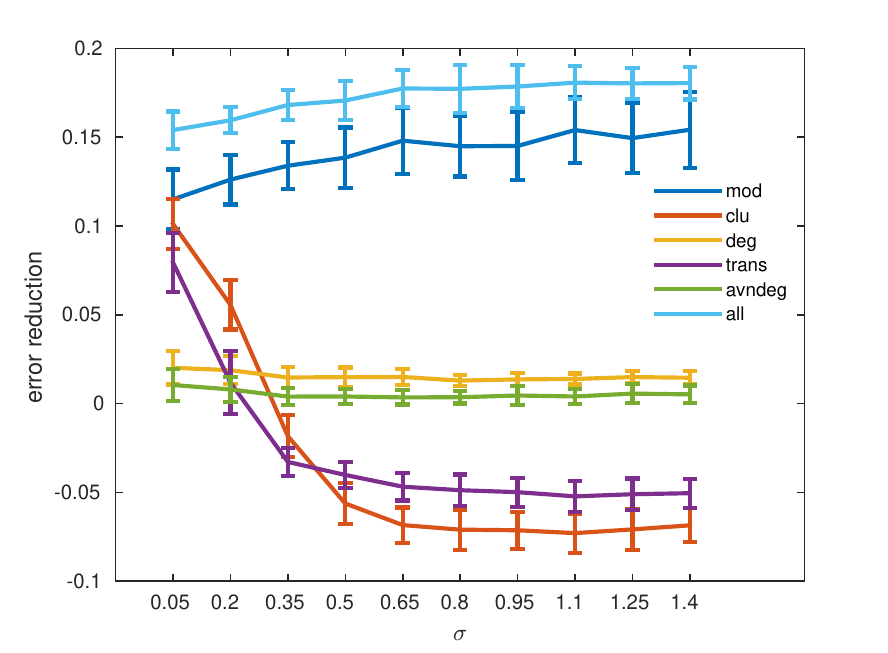}
	\caption{Denoising of a modular network with 128 nodes and 8 modules and $90\%$ of the weights in the modules. We show the error reduction for an increasing noise level and various graph metrics. Note that for the transitivity and clustering coefficient the estimated network had a larger error than without applying the denoising algorithm. This can be explained by noting that the weights are clustered in modules whereas the denoising algorithm is free to adapt any weight.}
	\label{fig:denmodvn}
\end{figure}

Next, we show the error reduction in terms of the number of nodes in the network and a noise level of $\sigma=.5$, see Figure \ref{fig:dennodes}.

\begin{figure}[htbp]
	\centering    

	\includegraphics[width=8cm]{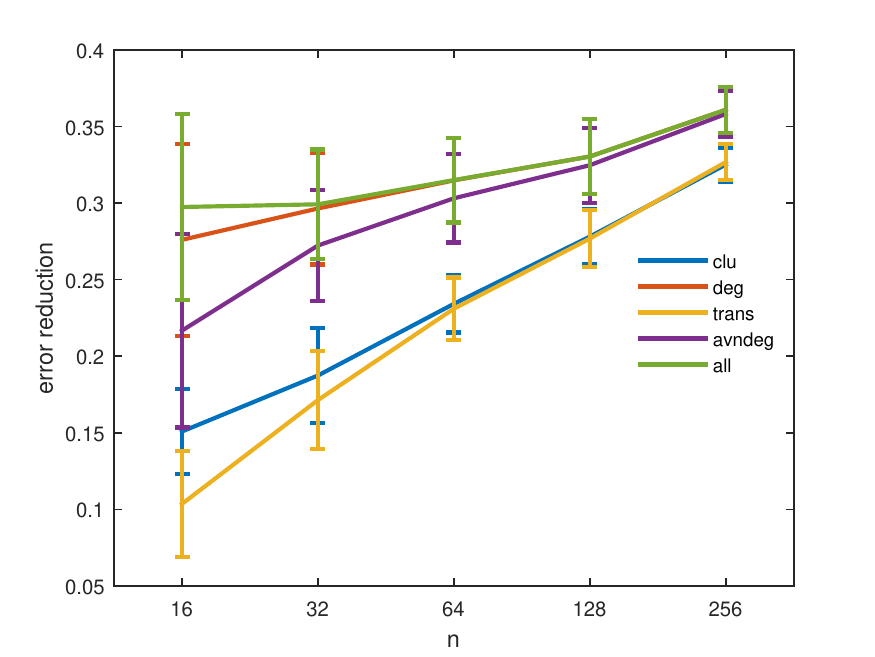}
	\caption{Effect of the number of nodes on the error reduction and a constant $\sigma=.5$. Although the error reduction increases with the number of nodes there is an increase in absolute error as well.}
	\label{fig:dennodes}
\end{figure}

%
%
%\begin{figure}[htbp]
%	\centering    
%	\includegraphics[width=8cm]{den1plot1.pdf}
%	\caption{Example denoising of a noisy network. We show the original, the noisy observed network and denoised estimates based on one and four metrics. Noise level is at $\sigma=.75$.}
%	\label{fig:den1plot1}
%\end{figure}
%\begin{figure}[htbp]
%	\centering    
%	\includegraphics[width=8cm]{den1plot2.pdf}
%	\caption{Example denoising of a noisy network. We show the original, the noisy observed network and denoised estimates based on one and four metrics. Noise level is at $\sigma=1.5$.}
%	\label{fig:den1plot2}
%\end{figure}

\subsubsection{Real EEG data}

The denoising algorithm was applied on two electroencephalography (EEG) datasets on a memory task from Alzheimer's patients and control subjects \cite{Pietto2016}. In dataset-1, there were 128-channel recordings from 13 patients with mild cognitive impairment (MCI) and 19 control subjects while for dataset-2 there were 64-channel recordings from 10 patients with familial Alzheimer's disease (FAD) and 10 control subjects. For both datasets we selected the common subset of 28 electrodes that have the exact same locations on the scalp. For each subject we split the recording into 1 second epochs and computed the connectivity network matrix of the first 50 epochs and of each epoch separately. We used the imaginary part of coherence as connectivity metric and considered the alpha (8-12Hz) spectral band. 

We tested two settings. Firstly, a train-test scheme where connectivity matrices are obtained for the two datasets separately and the weight matrices of the test set (dataset-1, MCI) are denoised according to the graph metrics of the training dataset (dataset-2).  Secondly, a within-subject denoising setting where the connectivity matrix of the first 50 epochs of a subject (from dataset-2, FAD) is considered as the training weight matrix $\tb{W}_{tr}$, and each of the other matrices is a test matrix $\tb{W}_i$ to be denoised. Here, we used the transitivity and degree as graph metrics combining both the global and local properties of brain activity.

In Table I we show the results of the denoising algorithm when using as target values the graph metrics of a) the same subject, b) the same group (patient/control), c) the opposite group. The mean square error (MSE) is computed against the training network $\tb{W}_{tr}$. Differences between subjects were significant under a unpaired ttest, $p<0.01$, for all pairwise comparisons except between the same group and opposite group. In Figure \ref{fig:den2afava} we show an example of denoised networks for one trial of a subject. 
\begin{table}\caption{Denoising of noisy EEG networks of 28 channels. For each of the 32 test subjects (19 controls - 13 patients from dataset-1) we denoised each trial of the subject's weight matrix based on the graph metrics of the mean weight matrices of the two groups from the training set (dataset-1) and its own weight matrix. We show the average MSE over all subjects for the three denoising schemes and the original noisy network. Error is computed is against $\tb{W}_{tr}$.}
\small{$
\begin{array}{cccc}\hline \vspace{2pt}
\text{Same Group}& \text{Opposite Group}  & \text{Self}  & \text{Noisy original} \\ \hline 
5.31 \pm 0.79  & 5.33 \pm 0.88  & 4.51\pm0.66  & 6.26\pm1.21
\vspace{2pt} \end{array} $}
\end{table}

\begin{figure}[htbp]
	\centering    
	\includegraphics[width=9cm]{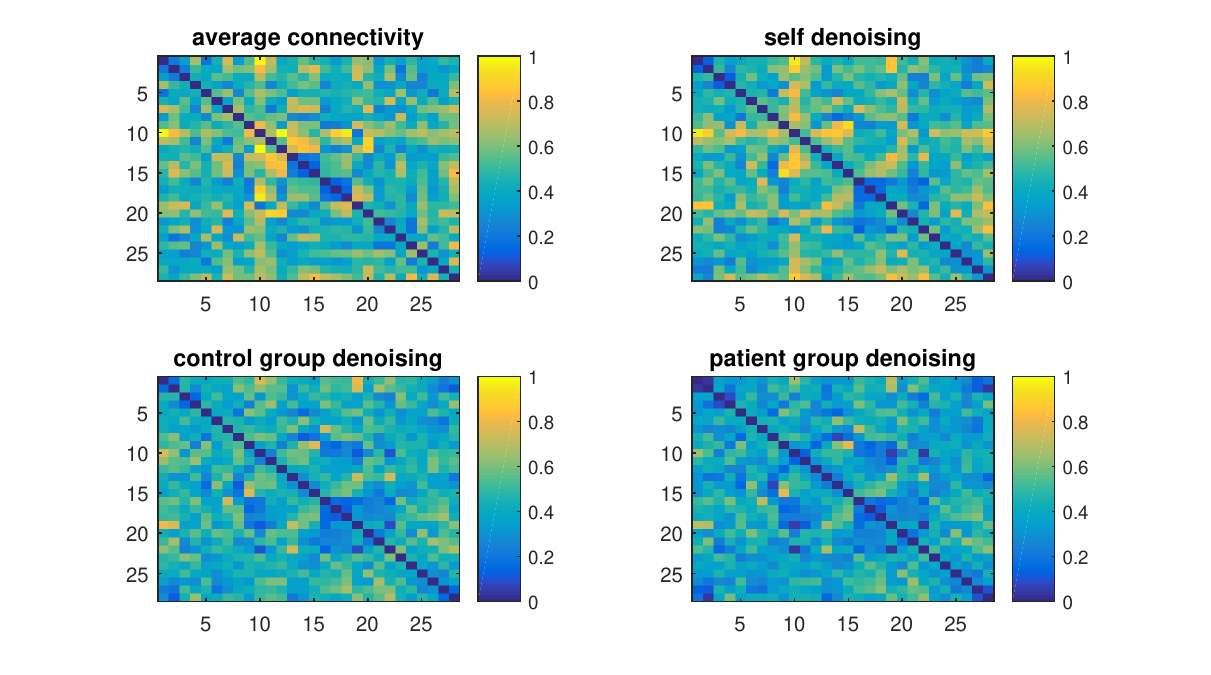}
	\caption{Example of obtained denoised networks for one a trial of a single subject. The network denoised based on the subject specific graph metrics retained the basic characteristics of the true network.}
	\label{fig:den2afava}
\end{figure}

\subsubsection{Macaque connectivity}

In order to further demonstrate the potential applicability of the algorithm, we obtained the connectivity matrices from the brains of two Macaques \cite{nr} with similar values for the transitivity metric. In Figure \ref{fig:den2maca} we show the result of the denoising where the first Macaque's (M1) brain was used to estimate the transitivity and drive the denoising for the second Macaque (M2).

\begin{figure}[htbp]
	\centering    
	\includegraphics[width=9cm]{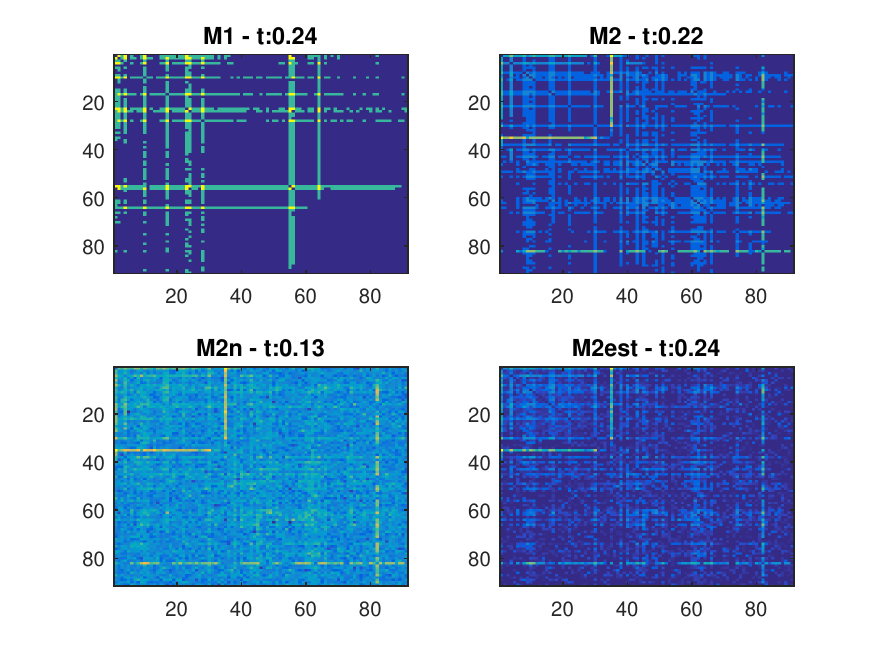}
	\caption{Example of obtained denoised networks when using M1's data to estimate the transitivity and apply that to the network of M2. It is observed that the denoised network is close to the true network even though the transitivity was set to that of M1.}
	\label{fig:den2maca}
\end{figure}

\subsection{Completion}

\subsubsection{Synthetic Networks}

Here we show the results of the network completion Algorithm \ref{alg:completion} for an increasing number of missing entries and number of nodes. Each separate case considers the average of 50 noise matrix realisations.Here we deal with a random network with $w_{ij} \sim \mc{U}[0,1]$. For each case, we optimise three graph metrics (transitivity, degree, clustering coefficient) and show the error reduction of the completion procedure. The missing values are initialised to 0.5 since this is the mean of the uniform $\mc{U}[0,1]$ distribution. This initialisation produces the smallest distance to the true network from all possible initialisations. Error reduction is calculated the same way as in Eq. \ref{eq:erred} with $\tb{W}_e$ being the initialised network.
Results are shown in Figure \ref{fig:comp1}.
\begin{figure}[htbp]
	\centering    
	\includegraphics[width=8cm]{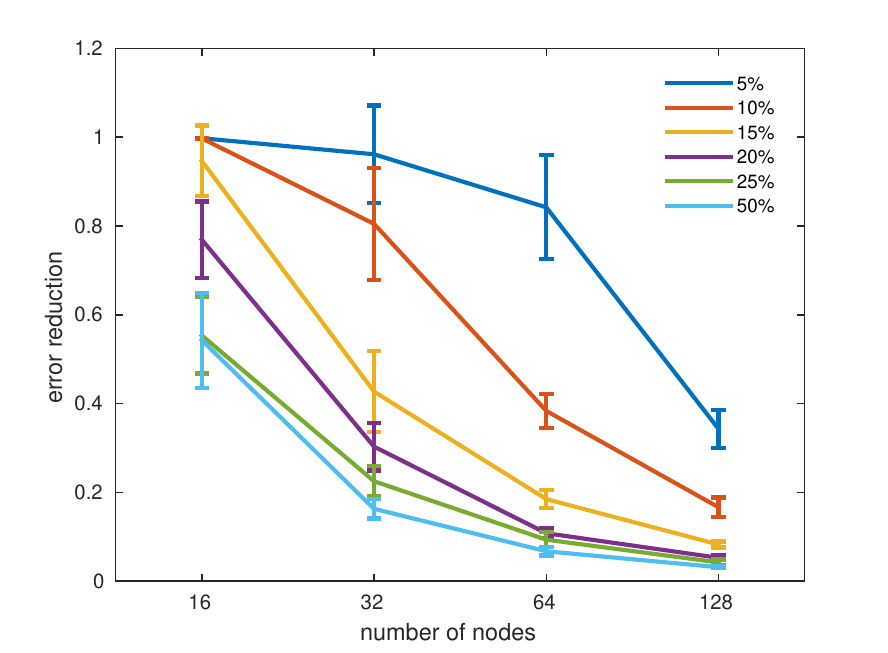}
	\caption{Error reduction of the completion procedure for different number of nodes and percentage of missing entries. We use here three graph metrics (degree, transitivity, clustering).}
	\label{fig:comp1}
\end{figure}

\subsubsection{Real Networks}

\paragraph{Known metrics} We applied the completion algorithm on the following datasets. 1) USAir: a 330 node network of US air transportation, where the weight of a link is the frequency of flights between two airports \cite{pajek}. 2) Baywet: a 128 node network which contains the carbon exchanges in the cypress wetlands of south Florida during the wet season \cite{koblenz}. The weights indicate the feeding levels between taxons. 3) Celegans: the neural network of the  worm C. elegans. Nodes indicate neurons and the edges are weighted by the number of synapses between the neurons \cite{Watts1998}. In Figure \ref{fig:compreal} we show the results of the completion algorithm \ref{alg:completion} for all three networks and different graph metrics.
\begin{figure*}[htbp]
	\centering    
	\includegraphics[width=18cm,height=6cm]{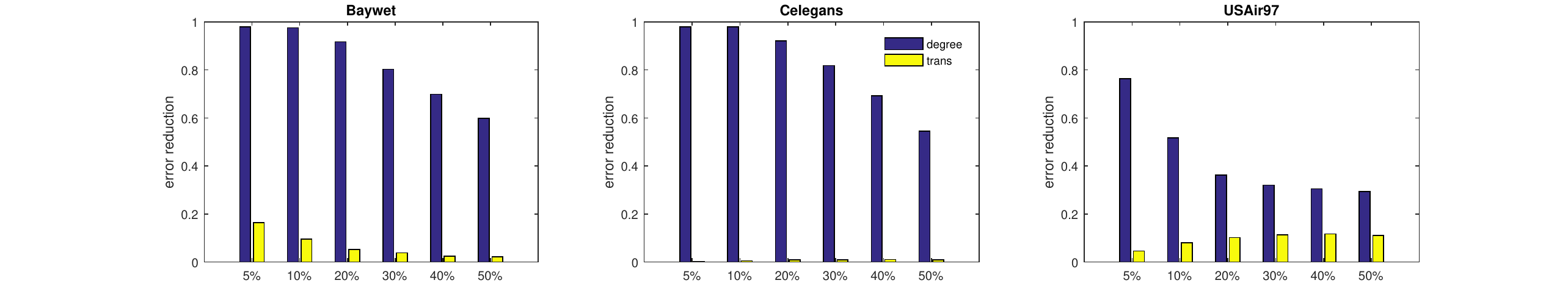}
	\caption{Error reduction of the completion procedure for the three networks in terms of the percentage of missing entries. The estimated network was closer to the true network for both global (transitivity) and local (degree) metrics.}
	\label{fig:compreal}
\end{figure*}
\paragraph{Uncertain metrics}
In this section we demonstrate the performance when the network metrics come from different datasets. This showcases a realistic scenario when a similar type of network is utilised in order to complete the missing values. In Figure \ref{fig:compunc} we show the results of the completion (a) on the Baywet dataset with the metric obtained in the Baydry dataset which contains the carbon exchanges in the dry season \cite{nr} and (b) on two enzyme network with the metrics estimated from the g355 enzyme and the completion performed on the g504 enzyme \cite{nr}.
\begin{figure}[htbp]
	\centering    
	\includegraphics[width=9cm,height=6cm]{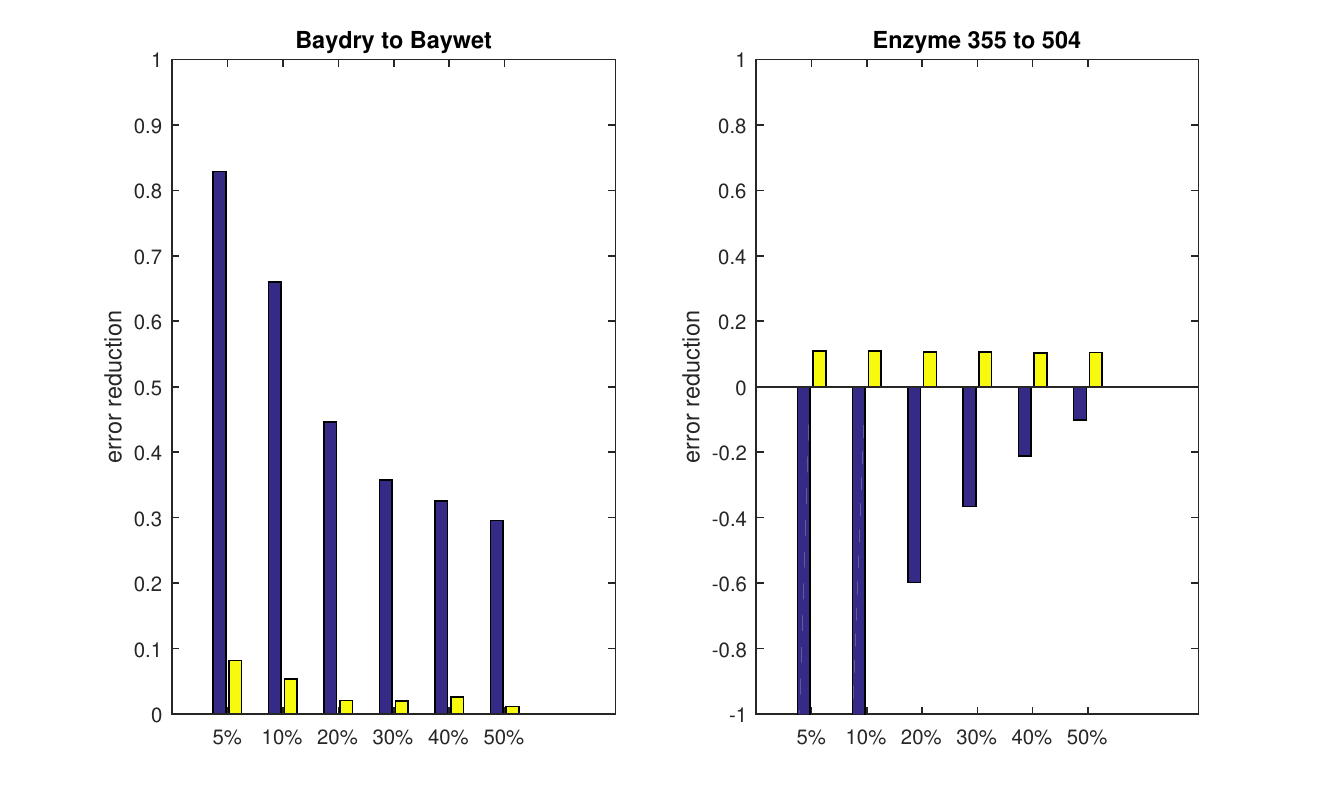}
	\caption{Error reduction of the completion procedure when using network metrics from different sources. For the Baywet dataset, the estimated network was closer to the true network for both global (transitivity) and local (degree) metrics. On the other hand, for the enzyme networks only the global metric resulted in a decrease in the error due to the large individual differences.}
	\label{fig:compunc}
\end{figure}
\subsection{Decomposition}
\subsubsection{Synthetic Networks}
In this section we show example results of the decomposition scheme by considering by mixing (as their sum) a modular and scale-free network. The modular network consists of 8 modules. For the modular network we use the modularity and  for the scale free network we use the transitivity as graph metrics to optimise. In Table II we show the average error reduction of the two networks of Algorithm \ref{alg:alternating} as compared with only denoising the two networks separately. In this cases error reduction for a single network is defined as:
\begin{equation}\label{eq:erred}
er=1-\frac{||{\tb{W}_{dec}}-\tb{W}||_{\text{F}}}{||\tb{W}_{den}-\tb{W}||_{\text{F}}}
\end{equation} 
For all cases the penalty parameter $\lambda$ is adjusted such that the reconstruction error is reduced w.r.t. iterations of the alternating minimisation procedure. 

\begin{table}\begin{center}\label{tab:decfurther2}\caption{Further error reduction of Algorithm \ref{alg:alternating} compared to only denoising the separate networks. In this case we have mixed a modular with a scale-free network. With this figure we illustrate that the optimisation problem of Eq. \ref{eq:deccon1} results in reducing the error between the estimates and the true networks as would happen from problem in Eq. \ref{eq:explanation}.}
	\small{$
	\begin{array}{c|cccc}\hline 
\text{Nodes} & \text{16}& \text{32}  & \text{64}  & \text{128} \\  \hline 
	\text{Reduction} & 0.35\pm 0.07  & 0.30\pm 0.24  & 0.24\pm0.30  & 0.22\pm0.22
	\vspace{2pt} \end{array} $}
\end{center}\end{table}
\subsubsection{Airline data}
 The dataset is a binary network that contains the networks for 37 airlines and 1000 airports \cite{cardillo2012}. We converted each airline's  binary network to a weighted network by adjusting any existing edge between two airports to the total number of edges for all airlines. The mixed network consists of two networks of different airlines mixed together (Lufthansa, Ryanair) for a subset of 50 nodes (airports). In Figure \ref{fig:dec2airplane} we show the decomposition result by using only global graph metrics (transitivity and global clustering coefficient). 

\begin{figure}[!htbp]
  
	\includegraphics[width=9cm,height=10cm]{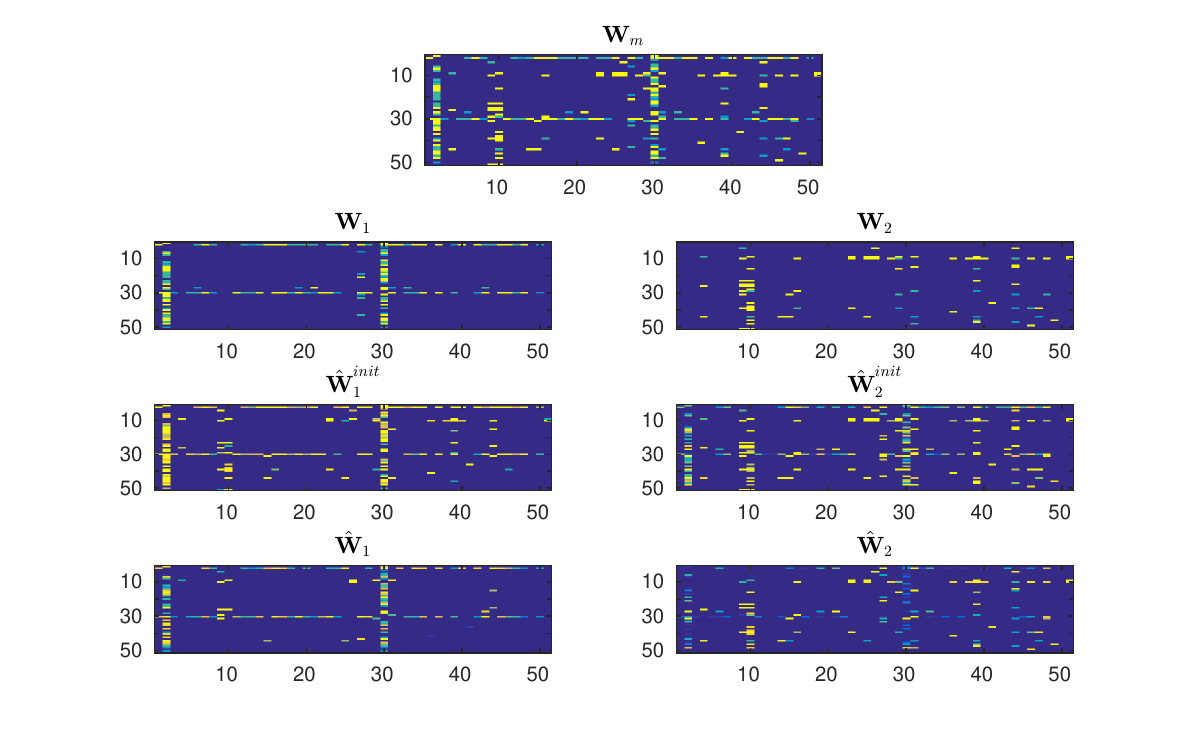}
	\caption{Decomposition of a mixed airline network $\tb{W}_f$ of 50 nodes into estimates $\hat{\tb{W}}_1$ and $\hat{\tb{W}}_2$. The true networks are shown as $\tb{W}_1$ and $\tb{W}_2$. The initial estimates of the denoising algorithm for the two networks are shown as $\tb{W}_1^{init}$ and $\tb{W}_2^{init}$. In this case we have used only global metrics (transitivity and clustering coefficient).}
	\label{fig:dec2airplane}
\end{figure}

\section{Discussion}\label{sec:discussion}

The optimisation schemes described in this work enable the adjustment of a network's weight matrix to fulfil specific properties. The utility of the denoising scheme (Section \ref{sec:denoising}) was evaluated on a number of cases including real-world data. It has to be pointed out that for convex graph metrics, the denoising scheme is guaranteed to converge to a network that is closer to the true underlying network than the noisy observed network. In Figure \ref{fig:denrandomvn} it is observed that considering the degrees of the nodes overshadows any other global metric's performance. For non convex graph metrics there is no guarantee but as shown in Figures \ref{fig:denrandomvn} and \ref{fig:denscalevn}, the estimated network is a better estimate of the underlying network than the original noisy version, even for increasing noise and different network types. The only exceptions to this can be observed for the clustering metrics (transitivity, clustering coefficient) of the modular network in Figure \ref{fig:denmodvn}. This can be explained by noting that in our example most of the weights ($90\%$) are clustered in modules while the denoising algorithm operates on all the weights. Constraining the weight updates only in the modules alleviates that problem. The utility of this scheme is also displayed in Figure \ref{fig:dennodes} where an increase in the number of nodes does not affect the performance. We point out that although the error reduction increases as the number of nodes increases, the error actually increases in absolute terms. 

We also tested the efficacy of the scheme in a real world EEG connectivity dataset of Alzheimers patients and control subjects. When performed in a within subject fashion, and splitting each subject's data into a train set to obtain estimates of the graph metrics, and a test set to apply the denoising algorithm, we successfully reduced the variability of the network regarding background EEG activity. More importantly, though, in a leave-subject-out procedure; using other subject's graph metrics as prior knowledge the algorithm was able to decrease the noise of the network, albeit not as much as in the within subject paradigm as expected. This has important implications in the EEG and related fields (e.g. BCI, fMRI, DTI) where subject independent paradigms are necessary to obtain practically usable and consistent performance \cite{Spyrou2015a}. This is supported by the Macaque example where the transitivity metric of one Macaque was used to drive the denoising procedure for the connectivity matrix of another Macaque, see Figure \ref{fig:den2maca}. Furthermore, such an approach can be useful in any application that can obtain prior knowledge of the structure of the network under consideration. Our work extends prior work in the network denoising field \cite{Gao2013,Morris2004,Aghagolzadeh2015} by providing theoretical proofs and empirical evidence that it is viable for a variety of network types. We also provide the strong proof that the degree of a network, resulting in a convex cost function, is very important in network estimation.

Weighted network completion has not been attempted in the literature and here we employed graph metrics as the driving force behind estimating the missing weights. For modest sizes of missing entries we showed that there is significant benefit of using the completion Algorithm \ref{alg:completion} up to 128 node networks and using three graph metrics. Similarly, for real networks, we show that the knowledge of graph metrics can aid in completing the network. More importantly, there is a benefit from the knowledge of global metrics as seen in Figure \ref{fig:compreal}. For the cases that the metrics are not completely known we show the efficacy of the completion algorithm in two cases where we use similar datasets to drive the completion procedure, see Figure \ref{fig:compunc}.

When utilising graph metrics in the network decomposition scheme (Section \ref{sec:decomposition}) the first observation is that the error reduction is greater by using information from both networks that combine to produce the mixed network. As shown in Table II the resulting estimates from the alternating minimisation procedure of Algorithm \ref{alg:alternating} confirm the intuition behind that methodology. It has to be noted that this is expected since more information is used as compared to only denoising. Namely, that the constraints employed in the optimisation problems of (\ref{eq:deccon1}) and (\ref{eq:deccon2}) can result in the behaviour of optimisation problem (\ref{eq:explanation}). Also in Table II, we show the error reduction as a function of the number of nodes. The choice of global graph metrics, modularity and transitivity, further demonstrates the utility of the methodology in setting where local information would be difficult to obtain.

The decomposition algorithm was tested on real data on mixed airline networks for both global and local only metrics. Assuming prior information on the structure of the networks, the algorithm was able to produce reliable estimates for the underlying networks, see Figure \ref{fig:dec2airplane}. Network decomposition for mixed networks has been rarely attempted in the literature. Our work provides a new formulation that takes account prior knowledge and the use of the reconstruction error in the estimation process. This is in contrast with \cite{Nash-Williams1964,Morris2004} where only factor graphs are considered. 

The choice of which graph measure should be used depends on the application and on which may be available. Local metrics assume knowledge of individual nodes' properties as is typical for e.g EEG applications (Figure \ref{fig:den2afava}) but may not be the case for the airport data (Figure \ref{fig:dec2airplane}). There are two limitations of this study that will be addressed in future work. Computational complexity and scaling the efficacy to large networks. Although many real world networks (e.g. EEG, Social Networks, Weather Networks, Airline Networks) are of the size that we consider in this work ($100s$ of nodes) there exist networks that consist of very large number of nodes ($>>1000s$ of nodes). The calculation of the derivatives increases in the order of $\mathcal{O}(n^2)$ making the computation slow and inefficient. As also discussed in \cite{shuman2013}, this is an open issue in graph based metrics. For large computational operations on graphs, polynomial approximations have been proposed \cite{hammond2011}. Similarly, as a network increases in size, the number of graph metrics should be increased in order to produce comparable performance to the medium sized networks considered here. This is because the number of unknown parameters to estimate increases therefore more graph metrics are necessary. In order to complement the popular graph metrics that we used in this work, we are considering other potential candidates such as ones based on graph spectral methods, and statistical graph analysis. 

\section{Conclusions}\label{sec:conclusions}

In this work we developed three mathematical optimisation frameworks that were utilised in network denoising, decomposition and completion. The basis of the methodology lies in adjusting a network's weights to conform with known graph measure estimates. We derived expressions for the derivatives of popular graph metrics and designed algorithms that use those derivatives in gradient descent schemes. We tested our proposed methods in toy examples as well as real world datasets.

The work performed here has the following implications for network estimation. Firstly, we showed that the use of graph metrics for network denoising reliably reduces the noise in an observed network for both convex and non-convex graph metrics. Also, by combining multiple graph metrics, further reduction is ensued. Depending on the type of network, some metrics may be more appropriate than others. Modularity works well for modular networks while degree seems to perform well for both random and scale free networks. For network decomposition, the use of global information as prior knowledge was sufficient to separate the underlying networks from their mixture. Such a framework can be the basis for constrained matrix or tensor decomposition of dynamic networks or multilayer networks. Finally, we provide a new weighted network completion paradigm that can complement existing matrix completion algorithms.

Other applications of our methodology can be weighted network reconstruction; the field that designs networks from scratch fullfilling specific criteria (e.g. a specific value for transitivity). The design of such network from scratch can be performed by transversing the level sets of a graph measure through the graph measure derivatives. Link prediction can also be incorporated by considering not only the weight similarities between different nodes but also the similarity between their derivatives.
\section*{Acknowledgment}

We would like to thank Dr Mario Parra Rodriguez (Heriot-Watt University) for making the EEGs available to us.

\bibliographystyle{IEEEtran}
\bibliography{../citations}
\begin{IEEEbiography}[{\includegraphics[width=1.in,height=1.25in,clip,keepaspectratio]{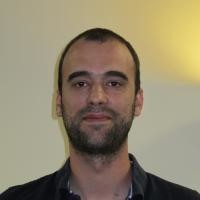}}] 
Loukianos Spyrou is a Research Associate at the Institute for Digital Communications, University of Edinburgh's School of Engineering. He is currently working on the MASNET project with Prof. John Thompson. His primary research interests are in machine learning and signal processing methodologies with focus on biomedical applications. Loukianos received the M.Eng. degrees from the University of York, Department of Electronics in 2004, the M.Sc. degree from the King's College London in 2005. His PhD was awarded in 2009 from Cardiff University. He has been working as a postdoctoral researcher since 2012 with previous posts in Radboud University and in the University of Surrey.
\end{IEEEbiography} 
\begin{IEEEbiography}[{\includegraphics[width=1.25in,height=1.25in,clip,keepaspectratio]{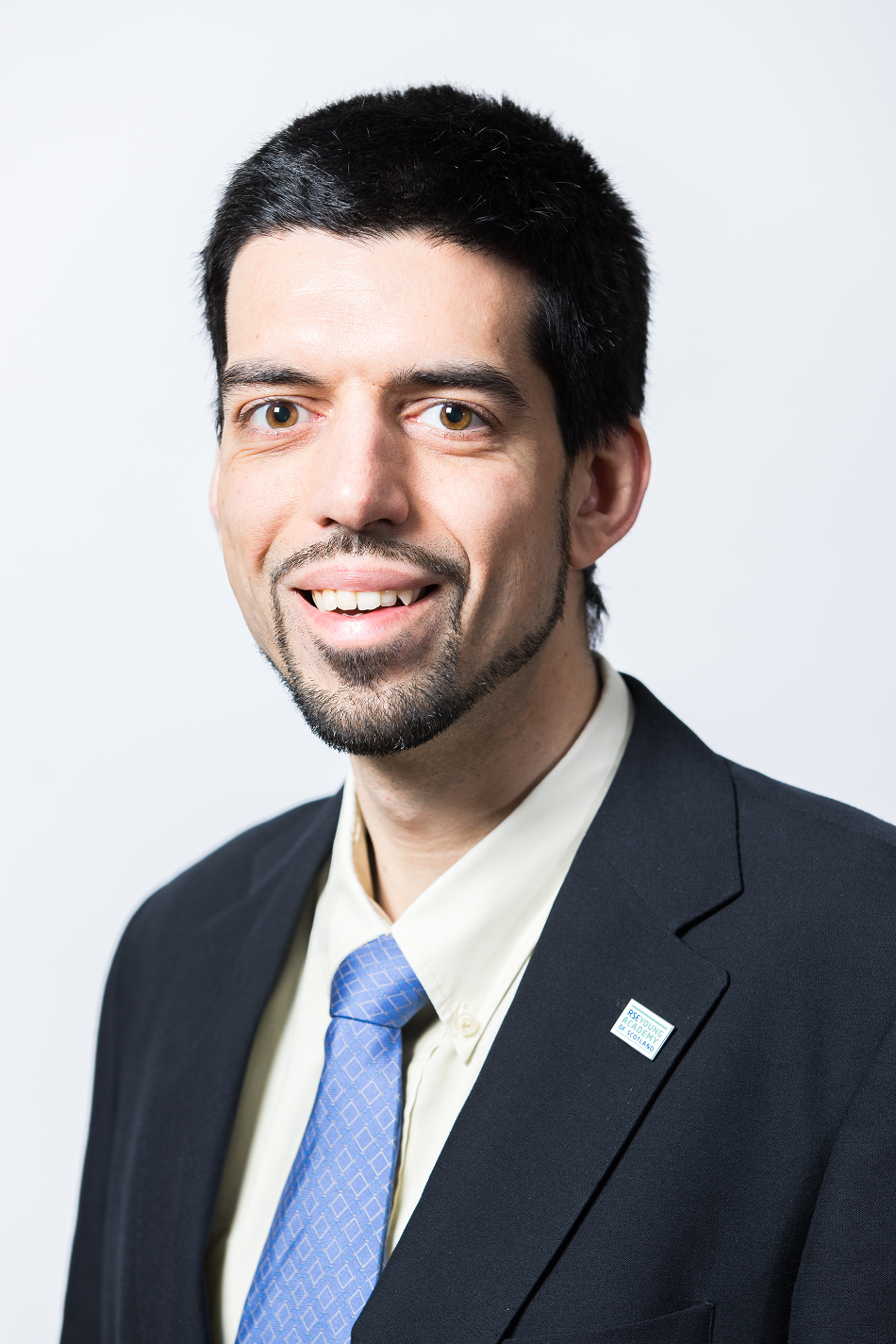}}]
Javier Escudero (S’07–M’10) received the MEng and PhD degrees in telecommunications engineering from the University of Valladolid, Spain, in 2005 and 2010, respectively. Afterwards, he held a post-doctoral position at Plymouth University, UK, until 2013. He is currently a tenured faculty member (Chancellor's Fellow) at the School of Engineering of the
University of Edinburgh, UK, where he leads a research group in biomedical signal processing with particular interests in non-linear analysis, network theory, and multiway decompositions.
He is author of over 45 scientific articles. Dr Escudero received the Third Prize of the EMBS Student Paper Competition in 2007 and the award to the best PhD thesis in healthcare technologies by the Spanish Organization of Telecommunications Engineers in 2010. In 2016, he was elected member of the Young Academy of Scotland. He is president of the society of
Spanish Researchers in the United Kingdom during the 2018/19 term.
\end{IEEEbiography} 
\end{document}

\newpage
\appendices
\section{}\label{app:proof}
\begin{theorem}
	Let $\tb{x},\tb{y},\tb{z} \in \mathbb{R}^n$ and $f(x)=M$ with $f: \mathbb{R}^n \to R $. Suppose $\tb{z}=\tb{x}+\tb{y}$. If the cost function $c(\tb{z})=(f(\tb{z})-M)^2$ is convex, then gradient descent with iteration index $k \in (0...K)$:
	\begin{equation}
	\tb{z}^{k+1}=\tb{z}^k-\lambda \nabla{c(\tb{z}^k)}
	\end{equation} 
	will converge to a point $\tb{z}^K$ such that:
	\begin{equation}
	||\tb{x}-\tb{z}^K||_2 \leq ||\tb{x}-\tb{z}||_2
	\end{equation}
	for any $\tb{z}$ and small enough $\lambda$.
\end{theorem}
\begin{proof} The vectors $\tb{x},\tb{y},\tb{z}$ denote the vectorised versions of the adjacency matrices $\tb{W}$. This theorem describes the situation that a network ($\tb{x}$) is corrupted by noise ($\tb{y}$) resulting in a noisy network ($\tb{z}=\tb{x}+\tb{y}$). The purpose of this proof is to show that gradient descent with convex $c(.)$ will converge to $\tb{z}^K$ which is a better, in terms of distance, estimate of $\tb{x}$ than $\tb{z}$. The proof does not try to show that gradient descent on convex functions achieves a global minimum; instead that the minimum achieved is a better estimate of the network than the original noisy network $\tb{z}$. Note that for graph metrics in general, there are infinite solutions for a matrix $\tb{W}$ that achieves a minumum point. 
	
For strictly convex functions, the proof is trivial since there is a unique minimum which corresponds to $\tb{x}$ only and gradient descent will converge to that value.
	
For both convex and strictly convex functions the proof is as follows. The sublevel sets, $L_{s}(c)=\left\{ \tb{z} \in \mc{R}^n\,\mid \,c(\tb{z})< s\right\}$ of a convex function are convex sets. Furthermore, $L_{a}(c) \subseteq L_{b}(c)$ for any $a\leq b$. For any point $\tb{z}^k$ with $c(\tb{z}^k)=L_{k}$, the negative gradient $-\nabla c(\tb{z}^k)$ forms a right angle with the level set at $L_k$ (by definition). Since the level sets are convex sets, the gradient update $\tb{z}^k-\lambda \nabla{c(\tb{z}^k)}$ leads to a point $\tb{z}^{k+1}$ such that the angle between $\tb{z}^{k+1}\tb{z}^k$ and $\tb{z}^{k}\tb{z}^*$ is acute in the triangle $(\tb{z}^k)(\tb{z}^{k+1})(\tb{z}^*)$ for any $\tb{z}^* \in L_{k}(c)$ (see Figure \ref{fig:proof}) and small enough step size. Since the optimum point $\tb{x}$ is also contained in $L_{k}(c)$ and the relation between the distances $||\tb{z}^*-\tb{z}^{k+1}||<||\tb{z}^*-\tb{z}^k||$, this implies that $||\tb{x}-\tb{z}^{k+1}||<||\tb{x}-\tb{z}^k||$ for any $\tb{z}$ and $\tb{k}$ provided that the step size is small enough.  Small enough in the sense that the gradient step must not cross the level set of $\tb{z}^*$.
\end{proof}

\begin{figure}[htbp]	
	\centering    
	\includegraphics[width=9cm]{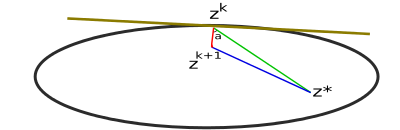}
	\caption{Triangle formed between $z^k$, $z^{k+1}$ and $z^*$. $z^*$ is any point inside the sublevel set at $f(z^k)=K_k$. For such a triangle the distance $||z^*-z^{k+1}||$ is always smaller than $||z^*-z^k||$. The closed curve denotes the boundary of the level set at $k$.}\label{fig:proof}
\end{figure}

\section{}
\subsection{Derivatives of scalar functions of matrices} \label{app:deriv}
\noindent
Differentiating a scalar function $f(\bf W)$ w.r.t. a matrix $\bf W$, $\frac {\partial f}{\partial \mathbf {W} }$, is essentially a collection of derivatives w.r.t. the separate matrix elements placed at the corresponding indices, i.e.: 
\begin{equation}
\frac {\partial f}{\partial \mathbf {W} }=
\begin{bmatrix}{\frac {\partial f}{\partial w_{11}}}&{\frac {\partial f}{\partial w_{12}}}&\cdots &{\frac {\partial f}{\partial w_{1n}}}\\{\frac {\partial f}{\partial w_{21}}}&{\frac {\partial f}{\partial w_{22}}}&\cdots &{\frac {\partial f}{\partial w_{2n}}}\\\vdots &\vdots &\ddots &\vdots \\{\frac {\partial f}{\partial w_{n1}}}&{\frac {\partial f}{\partial w_{n2}}}&\cdots &{\frac {\partial f}{\partial w_{nn}}}\\
\end{bmatrix}
\end{equation}
Expressing the derivatives in matrix form allows easy and scalable formulation of the derivatives irrespective of the number of entries. If the derivative of a specific element of $\bf W$, indexed by $(i,j)$, is required to be processed separately, this can be performed by selecting the same index of the derivative matrix. For example $\{\frac{df}{d\bf W}\}_{ij}=\frac{df}{dw_{ij}}$.

In the case of undirected networks, where the weight matrices are symmetric, the following adjustment needs to be made to ensure that the derivatives are themselves symmetric:
\\
\begin{equation}
\frac{df}{d{\bf W}}=\frac{\partial f}{\partial {\bf W}}+\left(\frac{\partial f}{\partial {\bf W}}\right)^T-diag\left(\frac{\partial f}{\partial {\bf W}}\right)
\end{equation}
The formulations in the text are in terms of the partial derivatives for simplicity. 

\subsection{Modularity derivative} \label{app:mod}
The modularity is written as:

\begin{equation}
M=\frac{1}{l^w}\sum\limits_{ij}\left(w_{ij}-\frac{k_i^w k_j^w}{l^w}\right)\delta_{ij}
\end{equation}
\\
where $\delta_{ij}=1$ whenever nodes $i$ and $j$ belong to the same module and zero otherwise. The term $\frac{1}{l^w}\sum\limits_{ij}(w_{ij}\delta_{ij})$ can be written as:
\\
\begin{equation}
m_1=\frac{1}{l^w}\sum\limits_{ij}(w_{ij}\delta_{ij})=\frac{tr\{\bf{W\Delta^T}\}}{tr\{\bf WO_n\}}=\frac{\theta}{l^w}
\end{equation}
\\
where $\bf{\Delta}$ is the matrix that contains the $\delta_{ij}$. The gradient is given by:
\\
\begin{equation}
\frac{\partial m_1}{\partial \bf W}=\frac{ l^w{\bf\Delta}-{\theta \bf O_n^T} }{(l^w)^2}
\end{equation}
\\
The other term can be written as:
\\
\begin{eqnarray}
m_2=\frac{1}{(l^w)^2}\sum\limits_{ij}k_i^w k_j^w \delta_{ij}=\frac{1}{(l^w)^2}\sum\limits_{ijkl} w_{ik} w_{jl} \delta_{ij}=\\
\frac{1}{(l^w)^2}{\sum\limits_{r=1}^ntr\{\bf{W^T}\bf{C}_r\bf{W} \Delta^T \}}=\frac{1}{(l^w)^2}\sum\limits_{r=1}^n \xi_r
\end{eqnarray}
\\
where $\bf{C_r}$ is a circular shift matrix that shifts down the rows of the matrix on the right by $r-1$. 

\section{} \label{app:conv}
\begin{theorem}
	For any weighted network with a graph measure of the type $f(\tb{W})=tr\{\tb{W}\tb{A}\}$ the function $g(\tb{W})=\left(f(\tb{W})-K\right)^2$ is convex for any  matrix $\tb{A}$. 
\end{theorem}
\begin{proof}
In order to show that $g(\tb{W})$ is convex it suffices to show that the Hessian of $g$ is positive semidefinite. We will use differential notation \cite{Magnus1989} to calculate the Hessian which is defined as:
	\begin{equation}
	 \{\tb{H}g\}_{ij} \equiv \frac{\partial^{2} g}{\partial x_{i} \partial x_{j} } 
	\end{equation}
where $x_i$ is an element of the vectorized weight matrix $\tb{W}$. The first differential of $g(\tb{W})$ is:
\begin{equation}
dg(\tb{W})=2\left(tr\{\tb{WA}\}-K\right)tr\{d\tb{W}\tb{A}\}
\end{equation}
The second differential is:
\begin{equation}
d^2g(\tb{W})=2tr\{d\tb{W}\tb{A}\}tr\{d\tb{W}\tb{A}\}
\end{equation}
Using the relation between the trace and the vec operator, i.e. $tr\{\tb{A}^T\tb{B}\}=vec(\tb{A})^Tvec(\tb{B})$ and the circular property of the trace, i.e. $tr\{d\tb{W}\tb{A}\}=tr\{\tb{A}d\tb{W}\}$:
\begin{eqnarray}
d^2g(\tb{W})=2vec(\tb{dW}^T)^T vec(\tb{A}) vec(\tb{A}^T)^T vec(\tb{dW})\\
=2vec(\tb{dW})^T \tb{K}_{nn}vec(\tb{A}) vec(\tb{A}^T)^T vec(\tb{dW})\\
=vec(\tb{dW})^T 2vec(\tb{A}^T) vec(\tb{A}^T)^T vec(\tb{dW})
\end{eqnarray}
where $\tb{K}_{nn}$ is the commutation matrix satisfying $vec(\tb{X}^T)=\tb{K}_{nn} vec(\tb{X})$. Note that we have `commuted' $\tb{K}_{nn}$ from $vec(\tb{dW}^T)$ to $vec(\tb{A}^T)$. The second differential was brought to the form $d^2 g = vec(d\tb{W})^T \tb{Z} vec(d\tb{W})$ which means that the Hessian is \cite{Magnus1989}:
\begin{equation}
\tb{H}g = \frac{1}{2} (\tb{Z}+\tb{Z}^T)
\end{equation}
The matrix $\tb{Z}$ is of the type $\tb{Z}=\tb{c}\tb{c}^T$. Since $\tb{c}$ is a vector the product $\tb{c}\tb{c}^T$ is rank-one producing a single nonzero positive eigenvalue. Therefore $\tb{Z}$ is positive semi-definite and hence the Hessian is positive semi-definite.
\end{proof}